\documentclass[letterpaper, 10 pt, conference]{ieeeconf}  %

\IEEEoverridecommandlockouts                              %

\title{\LARGE \bf
On the Dual Implementation of Collision-Avoidance Constraints in Path-Following MPC for Underactuated Surface Vessels
}

\author{Simon Helling$^{1}$, Christian Roduner$^{2}$, and Thomas Meurer$^{1}$%
\thanks{$^{1}$Simon Helling and Thomas Meurer are with the Chair of Automatic Control, Faculty of Engineering, Kiel University, 24143 Kiel, Germany
    {\tt\small \{sh, tm\}@tf.uni-kiel.de}}%
\thanks{$^{2}$Christian Roduner, is with AVL Software and Functions GmbH, 93059 Regensburg, Germany
    {\tt\small \ christian.roduner@avl.com}}%
}

\usepackage{latexAbbreviations}
\usepackage{siunitx}
\usepackage{graphicx}
\usepackage{nicefrac}
\graphicspath{{figures/}}
\usepackage[indention=0pt,margin=0pt,skip=0pt,font=scriptsize]{subcaption}
\usepackage[font=footnotesize]{caption}
\usepackage{amsmath,amssymb}
\usepackage[dvipsnames]{xcolor}%

\usepackage{pgfplots}
\pgfplotsset{compat=newest}
\usetikzlibrary{patterns}
\usetikzlibrary{positioning}
\usetikzlibrary{arrows.meta}
\usetikzlibrary{external}
\usetikzlibrary{spy}
\tikzset{external/only named=true}

\usepackage{tabulary}
\usepackage{multicol}
\usepackage{multirow}
\usepackage{booktabs}

\newcommand{\vr}{\vec{r}}

\newcommand{\V}{\mathcal{V}}
\newcommand{\Om}{\mathcal{O}_{m}}
\newcommand{\vmum}{\vec{\mu}_{m}}
\newcommand{\vlambdam}{\vec{\lambda}_{m}}
\newcommand{\vzm}{\vec{z}_{m}}
\newcommand{\Am}{A_{m}}
\newcommand{\vbm}{\vec{b}_{m}} 
\makeatletter
\newcommand{\subalign}[1]{%
    \vcenter{%
        \Let@ \restore@math@cr \default@tag
        \baselineskip\fontdimen10 \scriptfont\tw@
        \advance\baselineskip\fontdimen12 \scriptfont\tw@
        \lineskip\thr@@\fontdimen8 \scriptfont\thr@@
        \lineskiplimit\lineskip
        \ialign{\hfil$\m@th\scriptstyle##$&$\m@th\scriptstyle{}##$\hfil\crcr
            #1\crcr
        }%
    }%
}
\makeatother
\usepackage{color}
\newtheorem{proposition}{Proposition}
\newtheorem{remark}{Remark}
\newtheorem{lemma}{Lemma}

\usepackage[absolute,showboxes]{textpos}

\TPMargin{5pt}

\newcommand{\copyrightstatement}{
    \begin{textblock}{0.87}(0.06,0.93)    % tweak here: {box width}(leftposition, rightposition)
        \noindent
        \textcopyright 2021 IEEE.  Personal use of this material is permitted.  Permission from IEEE must be obtained for all other uses, in any current or future media, including reprinting/republishing this material for advertising or promotional purposes, creating new collective works, for resale or redistribution to servers or lists, or reuse of any copyrighted component of this work in other works.
    \end{textblock}
}
\begin{document}
\maketitle
\copyrightstatement
\thispagestyle{empty}
\pagestyle{empty}

\begin{abstract}
    A path-following collision-avoidance model predictive control (MPC) method is proposed which approximates obstacle shapes as convex polygons. Collision-avoidance is ensured by means of the signed distance function which is calculated efficiently as part of the MPC problem by making use of a dual formulation. The overall MPC problem can be solved by standard nonlinear programming (NLP) solvers. The dual signed distance formulation yields, besides the (dual) collision-avoidance constraints, norm, and consistency constraints. A novel approach is presented that combines the arising norm equality with the dual collision-avoidance inequality constraints to yield an alternative formulation reduced in complexity. Moving obstacles are considered using separate convex sets of linearly predicted obstacle positions in the dual problem.  The theoretical findings and simplifications are compared with the often-used ellipsoidal obstacle formulation and are analyzed with regard to efficiency by the example of a simulated path-following autonomous surface vessel during a realistic maneuver and AIS obstacle data from the Kiel bay area.
\end{abstract}

\section{Introduction}
Collision-avoidance emerges as an essential problem for autonomous vessel operation and compliance to the international regulations for preventing collisions at sea (\mbox{COLREGs}), see \cite{Huang2020}. It is therefore crucial to construct efficient and robust implementations in order to achieve real-time feasible vessel trajectories. A powerful and widely used mathematical tool for this purpose is model predictive control (MPC) which sets itself apart from other nonlinear control approaches with its unique ability to handle input and state constraints. In this context, obstacles are usually approximated as ellipsoids, see, e.g., \cite{Eriksen2017}, \cite{Eriksen2020} and the references therein, which might not fit the requirements of autonomous operation in confined environments and lead to unnecessary or even infeasible (w.r.t. the environment) maneuvers. The ellipsoidal approach also fails to take into account the geometry of the controlled vessel. Consequently, a more flexible approach is to approximate obstacles as  CSG functions as in, e.g., \cite{Helling2020} or convex polygons which can be studied in combination with a fundamental concept in collision-avoidance, namely, the signed distance function. This can be expressed as
\begin{align}
    \label{eq:sec:introduction:definitionSignedDistance}
    \text{sd}(\V,\mathcal{O}) = \text{dist}(\V,\mathcal{O}) - \text{pen}(\V,\mathcal{O}),
\end{align}
where the sets $\V$ and $\mathcal{O}$ describe the geometry of the controlled vessel and an obstacle, respectively. Therein, {$\text{dist}(\V,\mathcal{O}):= \inf_{\vz}\{\norm{\vz}: (\V+\vz)\cap\mathcal{O}\neq\emptyset  \}$} describes the distance between the two sets
and
{$\text{pen}(\V,\mathcal{O}):= \inf_{\vz}\{\norm{\vz}: (\V+\vz)\cap\mathcal{O}=\emptyset  \}$} denotes the penetration depth. 
Collision-avoidance is ensured if 
\begin{align}
    \label{eq:sec:introduction:collisionAvoidanceCondition}
    \text{sd}(\V,\mathcal{O}) \geq d_{\text{safe}},
\end{align}
with an additional safety distance $d_{\text{safe}}\in\mathds{R}$.
In this context, several contributions have been made following different approaches, see, e.g., \cite{Schoels2020}, \cite{Schulman2013}, \cite{Borelli2020}. The work \cite{Schoels2020} is concerned with finding a convex inner approximation of the signed distance and a so-called action radius of a controlled robot and lays focus on kinodynamic constraints. In \cite{Schulman2013} a linearization of the signed distance function is used but approximation errors can lead to numerical difficulties. The contribution \cite{Borelli2020} utilizes concepts from convex optimization theory, see, e.g., \cite{Vandenberghe2004}, in order to transform the (primal) definition of the distance and penetration function to express the signed distance function in its dual form. 

In this contribution, the dual signed distance approach for convex polyhedra similar to \cite{Borelli2020}, which is based on \cite{Vandenberghe2004}, is revisited and a modification is introduced, which decreases the number of dual constraints in the MPC problem setup by combining constraints without altering the solution. The findings are embedded in a path-following setup, where an intuitive timing law is discussed that focuses on achieving convergence to the reference path defined by waypoints. The results are shown in a simulative study, where the different obstacle formulations are compared to each other in an exemplary manner. The simulation incorporates AIS data from the Kiel bay area. 

The paper is organized as follows. First, the vessel dynamics of the surface vessel model are given in Sec. \ref{sec:surfaceVesselModel}, where also the principles of modeling time-varying convex polyhedra is discussed. The contributions of \cite{Borelli2020} are briefly summarized in Sec. \ref{sec:dualSignedDistance} and the derivation of the proposed dual collision avoidance condition is given. Therein and throughout this contribution, we focus on the {full-body vessel} case, i.e., the vessel $\V$ is approximated as a convex polyhedron. However, this can be generalized to the case where $\V$ reduces to the vessel's center of origin (CO) $\vr$. In Sec. \ref{sec:mpcProblemFormulation} the path-following MPC problem is formulated, where the cost function and the timing law as well as the different obstacle constraint formulations are discussed in further detail. The latter distinguishes  between the ellipsoidal case, the (dual) polyhedral case according to \cite{Borelli2020}, and the proposed (dual) polyhedral case, which translates the conditions from Sec. \ref{sec:dualSignedDistance} into a set of obstacle constraints. The latter two formulations take into account the controlled vessel as a polygon while the ellipsoidal case only allows for the controlled vessel's CO to be considered. Simulation results are presented in Sec. \ref{sec:results}, including a comparison of the different implementations. Final remarks are provided in Sec. \ref{sec:conclusion}.
\section{Surface Vessel Model}\label{sec:surfaceVesselModel}
A three degrees of freedom (3DOF) model is used to perform the path following task. The 3DOF model utilizes two sets of coordinates. The first set, $\transpose{\vec{\eta}}=[\transpose{\vec{r}}, \, \psi ]$ describes the vessel's pose in the North-East-Down (NED) frame, where $\transpose{\vec{r}}=[x, \, y]$ is the position of the vessel's CO, where $x$ corresponds to the north and $y$ to the east coordinate. The third component $\psi$ describes the vessel orientation w.r.t. to north axis (heading). This set of coordinates is a reference frame for the second set $\tilde{\vnu}^\top = [ \tilde{u},\, v, \, r]$, which represents the vessel's surge and sway velocities as well as its yaw rate in a body-fixed coordinate frame, respectively. Therein, $\tilde u = u - u_0$ denotes the difference of the surge velocity from the nominal service speed with $u_0=\text{const.}$
\subsection{Dynamics}
The vessel dynamics can be represented using
\begin{align*}
\dot{\vec{\eta}}&=R_z(\psi)(\transpose{[u_0, \ 0,\ 0]}+\tilde\vnu)\\
M\dot{\tilde\vnu}&=\vec{\tau}(\tilde\vnu,\delta),
\end{align*}
where $\vec{\tau}(\cdot)$ is the vector of nonlinear forces and moments acting on the vessel. These forces and moments are approximated using a Taylor series approximation up to third order. The rudder angle $\delta$ constitutes the control input. The matrix $R_z(\psi)\in\text{SO}(3)$ connects the two coordinate frame velocities and the matrix $M$ specifies the system's inertia. Further information on system parameters can be found in \cite{Chislett1965}. With this, the system can be described as a continuous-time ordinary differential equation (ODE) 
\begin{align}\label{eq:sec:surfaceVessel:generalVesselOde}
    \vxd=\vf(\vx,u_c),\ t>0, \ \vx(0)=\hat{\vx}_0    
\end{align}
with the states $\transpose{\vx}=[\transpose{\vec{\eta}} ,\,\transpose{\tilde\vnu} ]$, control input $u_c=\delta$, and the fixed initial state given a measured or estimated (real) current state $\vxhat_{0}$.
\subsection{Geometry}
\begin{figure}[tb]
    \centering
    \vspace*{.12cm}% fit margin requirements?!
    \includegraphics{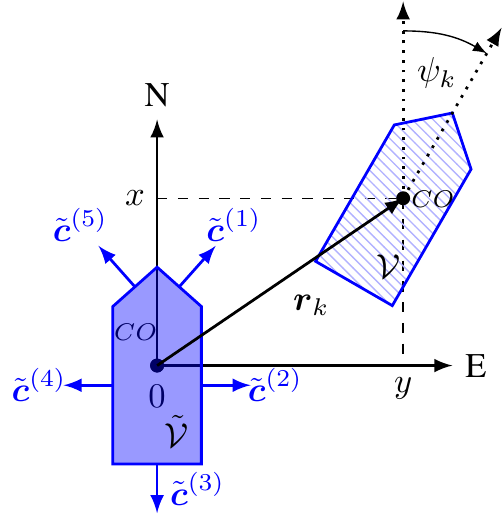}
    \captionsetup{belowskip=-.7cm}
    \caption{Base vessel shape $\tilde{\V}$ (blue) with outward normal vectors $\tilde{\vec{c}}^{(j)},\, j=1,\hdots,L$ for $L=5$ edges and vessel shape $\V(t_k)$ (striped blue) at $t_k$ in the pose $\vec{\eta}^\top_k=\vec{\eta}^\top(t_k)=[\vec{r}^\top_k,\,\psi_k]$ and with $C_k=C(t_k),\,\vec{d}_k=\vec{d}(t_k)$.}\label{fig:geometry}
\end{figure}

In the following, a base shape\footnote{In this context $\vx$ is a general vector in the two-dimensional space and not a system state as in, e.g., \eqref{eq:sec:surfaceVessel:generalVesselOde}. The different meanings will be clear from the context throughout this paper.} \mbox{$\tilde{\V} = \{\vx\inRn{2}\, \big\vert \, \tilde{C}\vx\leq\tilde{\vec{d}} \}$} is defined for the controlled vessel, which describes a convex polyhedron where the vessel's CO and the NED frame's coordinate origin coincide. Therein, $\tilde{C}=\transpose{[\tilde{\vec{c}}^{(1)}\hdots \tilde{\vec{c}}^{(L)}]}\inRnm{L}{2}$ consists of the (outward) normal vectors $\tilde{\vec{c}}^{(j)},\ j=1,\hdots,L$ defining the edges of the polyhedron, where $L$ is the number of edges needed for the controlled vessel and $\tilde{\vec{d}}\inRn{L}$ depends on the specific shape. This base shape is used to describe the controlled vessel shape's evolution in time, i.e.,
\begin{subequations}
    \label{eq:sec:surfaceVesselModel:controlledVesselShape}
    \begin{align}
    \V &= T(\tilde{\V})=\left\{ T(\vx)\, \vert \,\vx\in\tilde{\V} \right\}\\
    &= \left\{ \vx\inRn{2} \,\vert \,\tilde{C}T^{-1}(\vx)\leq \tilde{\vec{d}}\right\}\\
    \label{eq:sec:surfaceVesselModel:controlledVesselShapeFinal}
    &= \left\{ \vx\inRn{2}\, \vert\, C \vx \leq \vec{d} \right\},
    \end{align}
\end{subequations}
where the abbreviations $C = \tilde{C} \transpose{R}_z(\psi)$ and $\vec{d} = 	\tilde{\vec{d}}+\tilde{C}R_z^\top(\psi)\vec{r}$ and the affine transformations
\begin{align*}
    T&:\mathds{R}^2\mapsto\mathds{R}^2, \vx\mapsto R_z\left(\psi\right)\vx+\vec{r},\\
    T^{-1}&:\mathds{R}^2\mapsto\mathds{R}^2, \vx\mapsto \transpose{R_z}\left(\psi\right)\left(\vx-\vec{r}\right),
\end{align*}
have been used. For the sake of clarity explicit time dependencies are omitted. This transformation preserves the convexity of $\tilde{\V}$, see \cite{Vandenberghe2004}. See Fig. \ref{fig:geometry} for an illustration of the two sets. 

Analogously to \eqref{eq:sec:surfaceVesselModel:controlledVesselShape}, the base shape of any obstacle $m=1,\hdots,M$ can be described as a convex set $\tilde{\mathcal{O}}_m=\{\vx\inRn{2}\, \vert \, \tilde{A}_m\vx\leq\tilde{\vec{b}}_m \}$, with $\tilde{A}_m=[\tilde{\vec{a}}_{m}^{(1)}\, \hdots\, \tilde{\vec{a}}_{m}^{(L_m)}]^\top\inRnm{L_m}{2}$, which consists of the (outward) normal vectors $\tilde{\vec{a}}_{m}^{(j)},\ j=1,\hdots,L_m$ defining the edges of the $m$-th obstacle as a polyhedron, where $L_m$ is the number of edges and $\tilde{\vec{b}}_m\inRn{L_m}$ depends on the specific shape. The time evolution of the $m$-th obstacle shape follows as 
\begin{align}
    \label{eq:sec:surfaceVesselModel:obstacleShape}
    \mathcal{O}_m=\left\{\vx\inRn{2}\,\big\vert \, A_m\vx\leq\vec{b}_m \right\},
\end{align}
where {$A_m = \tilde{A}_m \transpose{R}_z(\psi_m)$} and {$\vec{b}_m = \tilde{\vec{b}}_m+\tilde{A}\transpose{R}_z(\psi_m)\vec{r}_m$}, with $\vec{\eta}_m^\top=[\vec{r}_m^\top,\, \psi_m]$ being the pose of the $m$-th obstacle.

\section{Dual Signed Distance Calculation}\label{sec:dualSignedDistance}
\begin{figure}[tb]
    \centering
    \vspace*{.12cm}% fit margin requirements?!
    \includegraphics{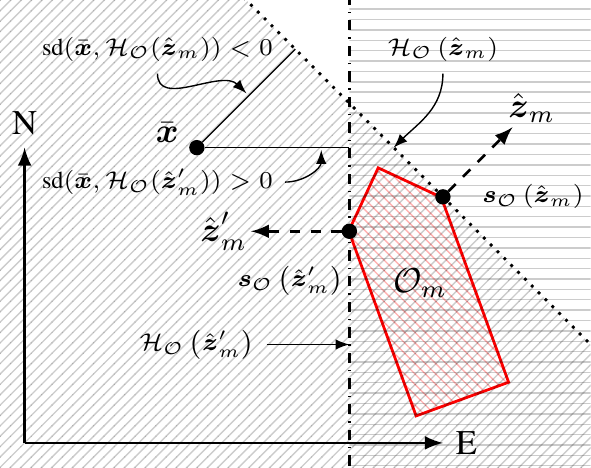}
    \captionsetup{belowskip=-.7cm} 
    \caption{Signed distances for a point $\bar{\vec{x}}$ and an obstacle $\mathcal{O}_m$ for two different vectors $\hat{\vec{z}}_m$ and $\hat{\vec{z}}_m^\prime$, as well as the corresponding support points $\vec{s}_{\mathcal{O}}(\cdot)$, supporting hyperplanes $\mathcal{H}_{\mathcal{O}}(\cdot)$ and lower supporting half-spaces $\mathcal{H}^-_{\mathcal{O}}(\cdot)$ (diagonal and horizontal gray lines). Note that in this example, $\text{sd}\left(\bar{\vec{x}},\mathcal{H}_{\mathcal{O}}(\vec{z}_m)\right)<0$ since $\bar{\vec{x}}\in\mathcal{H}^-_{\mathcal{O}}(\vec{z}_m)$ but $\text{sd}\left(\bar{\vec{x}},\mathcal{H}_{\mathcal{O}}\left(\hat{\vec{z}}_m^\prime\right)\right)>0$ since $\bar{\vec{x}}\notin\mathcal{H}^-_{\mathcal{O}}\left(\hat{\vec{z}}_m^\prime\right)$.}\label{fig:signedDistanceExample}
\end{figure}
In order to calculate the signed distance as defined by \eqref{eq:sec:introduction:definitionSignedDistance} efficiently, the dual approach to optimization-based collision-avoidance as shown, e.g., in \cite{Borelli2020} is extended. Therefore, the findings of \cite{Borelli2020} concerning the signed distance are briefly summarized. Here, we focus on the full-body vessel case, where both the controlled vessel and the obstacles are considered as fully-dimensional objects according to \eqref{eq:sec:surfaceVesselModel:controlledVesselShapeFinal} and \eqref{eq:sec:surfaceVesselModel:obstacleShape}, respectively. Furthermore, the original approach of \cite{Borelli2020} is extended and an alternative formulation is derived which potentially allows a more efficient calculation of the signed distance.
Since this contribution focuses mainly on results from convex optimization theory, the interested reader is referred to \cite{Vandenberghe2004} for further details and concepts on this topic. 
\subsection{Approach from \cite{Borelli2020}}
According to \cite{Borelli2020}, the collision avoidance condition as defined by \eqref{eq:sec:introduction:collisionAvoidanceCondition} can be expressed as\footnote{The results are adapted to the notation used in the present paper. Also, explicit time dependencies are omitted for the sake of clarity.}
\begin{align}
    \label{eq:sec:dualSignedDistance:resultsBorelli}
    \begin{split}
        &\text{sd}(\V,\mathcal{O}_m) \geq d_{\text{safe}} \\
        &\Leftrightarrow \exists\vlambdam\geq\vec{0},\vmum\geq\vec{0}:\\
        &\qquad\qquad -\transpose{\vec{\lambda}_m}\vec{d}-\transpose{\vec{\mu}_m}\vec{b}_m\geq d_{\text{safe}}\\
        &\qquad\qquad\norm{\transpose{A_m}\vec{\mu}_m}=1,\\ &\qquad\qquad\transpose{C}\vec{\lambda}_m+\transpose{A_m}\vec{\mu}_m=\vec{0},
    \end{split}
\end{align}
where $\vec{\lambda}_m\inRn{L}$ are (dual) variables associated with the controlled vessel for the $m$-th obstacle and $\vec{\mu}_m\inRn{L_m}$ are (dual) variables associated with the $m$-th obstacle. 
\subsection{Proposed approach}
For the derivation of the alternative approach, some Lemmata are in order.
\begin{lemma}[\cite{VanDenBergen2001}]\label{lemma:support_mapping}
    The support mapping $\transpose{\vzm}\vec{s}_{\mathcal{O}}(\vzm) = \sup_{\vx\in\mathcal{O}_m}\{\transpose{\vzm}\vx\}$, where $\vx\in\mathcal{O}_m\Leftrightarrow  \Am\vx\leq\vbm $,
    gives the maximum inner product for any $\vx\in\Om$ and thus defines the support point $\vec{s}_{\mathcal{O}}(\vzm)$ as the point in $\mathcal{O}_m$ which is furthest away from the origin in the direction of $\vzm$.
\end{lemma}
\begin{lemma}[\cite{Vandenberghe2004}]\label{lemma:convex_set_as_intersection_of_halfspaces}
    The convex set $\Om$ can be described using the intersection of all lower supporting half-spaces that contain it, i.e., $\Om=\bigcap_{\vzm} \mathcal{H}^-(\vzm,\vzm^\top\vec{s}_{\mathcal{O}}(\vzm))$, 
    with\footnote{In the following, the abbreviation $\mathcal{H}^-_{\mathcal{O}}(\vzm)=\mathcal{H}^-(\vzm,\vzm^\top\vec{s}_{\mathcal{O}}(\vzm))$ is used.} {$\mathcal{H}^-(\vzm,\vzm^\top\vec{s}_{\mathcal{O}}(\vzm))=\{ \vx\inRn{2} \,\vert\,  \transpose{\vzm}\vx \leq \transpose{\vzm}\vec{s}_{\mathcal{O}}(\vzm)\}$}.
\end{lemma}
\begin{remark}
    As can be seen, compared to \cite{Borelli2020}, the vector \mbox{$\vzm$}  in Lemma \ref{lemma:convex_set_as_intersection_of_halfspaces} does not need not be constrained in its length. This also becomes clear from the fact that $\vzm$ constitutes a normal vector of the supporting hyperplane \mbox{$
    \mathcal{H}_{\mathcal{O}}(\vzm)=\{ \vx\inRn{2}\, \vert \,   \transpose{\vzm}\vx = \transpose{\vzm}\vec{s}_{\mathcal{O}}(\vzm)  \}$} that is intrinsically connected to the lower supporting half-space $\mathcal{H}^-_{\mathcal{O}}(\vzm)$.
\end{remark}
\begin{lemma}[\cite{VanDenBergen2001}]\label{lemma:signed_distnace_minkowski_diff}
    The signed distance can be expressed as \mbox{$\text{sd}(\V,\Om) = \text{sd}(\vec{0},\Om-\V)$},
    where \mbox{$\Om-\V :=\{ \vx - \vy \, \vert \, \vx\in\Om,\ \vy\in\V \}$} describes the Minkowski difference.
\end{lemma}
\begin{lemma}[\cite{VanDenBergen2001}]\label{lemma:support_point_minkowski_diff}
    The support point $\vec{s}_{\mathcal{O}-\V}(\vzm)$ of the Minkowski difference $\Om-\V$ can be expressed as $\vec{s}_{\mathcal{O}-\V}(\vzm)=\vec{s}_{\mathcal{O}}(\vzm)-\vec{s}_{\V}(-\vzm)$.
\end{lemma}
\begin{proposition}\label{prop:sd_Vk_Okm}
    The collision-avoidance condition \eqref{eq:sec:introduction:collisionAvoidanceCondition} can be formulated with
        \begin{align}
            \label{eq:sec:dualSignedDistance:proposition}
            \begin{split}
                &\text{sd}(\vr,\Om) \geq d_{\text{safe}} \\
                &\Leftrightarrow \exists \vlambdam\geq\vec{0},\vmum\geq\vec{0}:\\
                &\qquad\qquad\frac{-\transpose{\vlambdam}\vec{d}-\transpose{\vmum}\vbm }{\norm{\transpose{\Am}\vmum}}\geq d_{\text{safe}}\\ 
                &\qquad\qquad\transpose{C}\vlambdam+\transpose{\Am}\vmum=\vec{0}.
            \end{split}
        \end{align}
\end{proposition}
\begin{proof}
    Consider the distance from an arbitrary point $\bar{\vx}\inRn{2}$ to the supporting hyperplane $\mathcal{H}_{\mathcal{O}}(\vz_m)$ defined by a normal vector $\vzm$ of an obstacle $\Om$ given by \eqref{eq:sec:surfaceVesselModel:obstacleShape}, i.e., 
    \begin{align*}
        \transpose{\hat{\vec{z}}_m}\vec{s}_\mathcal{O}(\vzm)-\transpose{\hat{\vec{z}}_m}\bar{\vx},
    \end{align*}
    where $\hat{\vec{z}}_m=\nicefrac{\vec{z}_m}{\norm{\vec{z}_m}_2}$ is the normalized vector with direction $\vec{z}_m$.
    In fact, this can be regarded as the negated signed distance between $\bar{\vx}$ and $\mathcal{H}_\mathcal{O}(\hat{\vz}_m)$, i.e.,
    \begin{align}
        \label{eq:sec:dualSignedDistance:signedDistanceXbarPartialH}
        \begin{split}
            \text{sd}(\bar{\vx},&\mathcal{H}_{\mathcal{O}}(\hat{\vz}_m)) \\
            &= 		\transpose{\hat{\vec{z}}_m}\left( \bar{\vx}-\vec{s}_{\mathcal{O}}(\hat{\vz}_m) \right)\\
            & = 
            \begin{cases}
                \leq 0, \quad \text{if\ } \bar{\vx}\in\mathcal{H}^-_{\mathcal{O}}(\hat{\vz}_m)\\
                > 0, \quad \text{else}
            \end{cases}
        \end{split}
    \end{align}
    see also Fig. \ref{fig:signedDistanceExample} where an example for two different vectors $\hat{\vec{z}}_m$ and $\hat{\vec{z}}_m^\prime$ is shown.
    The observation \eqref{eq:sec:dualSignedDistance:signedDistanceXbarPartialH} together with Lemma \ref{lemma:convex_set_as_intersection_of_halfspaces}
    can be combined to formulate the signed distance between the point $\bar{\vx}$ and the obstacle $\Om$ with 
    \begin{align}
        \label{eq:sec:dualSignedDistance:signedDistanceXbarOm}
        \text{sd}(\bar{\vx},\Om)=\sup_{\hat{\vec{z}}_m}\left\{ \text{sd}(\bar{\vx},\mathcal{H}_{\mathcal{O}}(\hat{\vec{z}}_m)) \right\},
    \end{align}
    or, in other words, it is the signed distance between the point $\bar{\vec{x}}$ and a particular supporting hyperplane for which \eqref{eq:sec:dualSignedDistance:signedDistanceXbarPartialH} yields the largest possible value.
    Subsequently, if the the convex set $\V$ is considered instead of the arbitrary point $\bar{\vec{x}}$, Lemma \ref{lemma:signed_distnace_minkowski_diff} and Lemma \ref{lemma:support_point_minkowski_diff} can be applied to give
    \begin{subequations}
        \begin{align}
            &\text{sd}(\V,\Om) = \text{sd}(\vec{0},\Om-\V)\\
            \label{eq:sec:dualSignedDistance:signedDistanceOriginMinkwoski}
            &=\sup_{\hat{\vec{z}}_m}\left\{ \text{sd}(\vec{0},\mathcal{H}_{\mathcal{O}-\V}(\hat{\vec{z}}_m)) \right\}\\
            &=\sup_{\hat{\vec{z}}_m}\left\{ -\transpose{\hat{\vec{z}}_m}\vec{s}_{\mathcal{O}-\V}(\hat{\vec{z}}_m) \right\}\\
            \label{eq:sec:dualSignedDistance:signedDistanceWithSupportPoints}
            &=\sup_{\hat{\vec{z}}_m}\left\{ -\transpose{\hat{\vec{z}}_m}\left( \vec{s}_{\mathcal{O}}(\hat{\vec{z}}_m)-\vec{s}_{\V}(-\hat{\vec{z}}_m) \right) \right\}.
        \end{align}
    \end{subequations}
    Comparing \eqref{eq:sec:dualSignedDistance:signedDistanceOriginMinkwoski} with \eqref{eq:sec:dualSignedDistance:signedDistanceXbarOm}, it can be observed that the signed distance between the two convex sets is equivalent to the signed distance between the origin and the Minkowski difference of the two sets. Incorporating Lemma \ref{lemma:support_mapping} for the arising support mappings in \eqref{eq:sec:dualSignedDistance:signedDistanceWithSupportPoints} yields
    \begin{align}
        \label{eq:sec:dualSignedDistance:signedDistancePrimal}
        \text{sd}(\V,\Om)&=
            \sup_{\hat{\vec{z}}_m}
                    \bigg\{\inf_{
                        \vx \in\mathcal{O}_m,\,
                        \vy\in\mathcal{V}}
                    \Big\{\transpose{\hat{\vec{z}}_m}(\vy-\vx)\Big\}\bigg\},
    \end{align}
    where $-\sup(\varphi)=\inf(-\varphi)$ and $\inf_A(\varphi)+\inf_B(\Phi) = \inf_{A\times B}(\varphi+\Phi)$ have been used, see \cite{Lojasiewicz1988}. With this, the Lagrange dual problem see, e.g., \cite{Vandenberghe2004} of the inner minimization in \eqref{eq:sec:dualSignedDistance:signedDistancePrimal} can be derived which leads to
    \begin{align}
        \label{eq:sec:dualSignedDistance:innerMinimizationDual}
        \begin{split}
        \text{sd}(\V,\Om)=
            &\sup_{
                        \vlambdam\geq\vec{0},\,
                        \vmum\geq\vec{0},\,
                        \vec{z}_m
                }
            \Bigg\{
                \frac{-\transpose{\vec{\lambda}_{m}}\vec{d}-\transpose{\vmum}\vbm}{\norm{\vec{z}_m}}:\\
                &-C^\top\vlambdam=\vec{z}_m,\,\Am^\top\vmum=\vec{z}_m
            \Bigg\},
        \end{split}
    \end{align}
    where $\sup_{x\in A}\left\{\sup_{y\in B}\varphi(x,y)\right\}=\sup_{A\times B}\left\{\varphi(x,y)\right\}$ has been used, see, e.g. \cite{Lojasiewicz1988}. Finally, eliminating $\vec{z}_m$ in \eqref{eq:sec:dualSignedDistance:innerMinimizationDual} yields
    \begin{align}
        \label{eq:sec:dualSignedDistance:signedDistanceResult}
        \begin{split}
            \text{sd}(\V,\Om)=&\sup_{\vmum\geq\vec{0},\,\vec{\lambda}_{m}\geq\vec{0}}\Bigg\{ \frac{-\transpose{\vec{\lambda}_{m}}\vec{d}-\transpose{\vmum}\vbm }{\norm{\transpose{\Am}\vmum}}:\\
            &\transpose{C}\vec{\lambda}_m+\transpose{A_{m}}\vmum=\vec{0} \Bigg\}.
        \end{split}
    \end{align}
    Since $\mathcal{O}_m$ and $\V$ are convex, strong duality holds. Thus, in order for the collision avoidance condition \eqref{eq:sec:introduction:collisionAvoidanceCondition} to hold, it suffices to find any feasible $\vlambdam,\vmum$ w.r.t. \eqref{eq:sec:dualSignedDistance:signedDistanceResult} such that 
    \begin{align*}
        \frac{-\transpose{\vec{\lambda}_{m}}\vec{d}-\transpose{\vmum}\vbm }{\norm{\transpose{\Am}\vmum}}\geq d_{\text{safe}}
    \end{align*}
    which concludes the proof.\footnote{Note that a similar reasoning can be applied to derive dual signed distance conditions when the controlled vessel is considered as the point given by $\vr$ only.}
\end{proof}

\section{MPC problem formulation}\label{sec:mpcProblemFormulation}
In the following, a soft-constrained path-following MPC will be discussed which can be realized by repeatedly solving optimal control problems on a receding horizon, which in turn can be written as
\begin{subequations}
    \label{eq:sec:mpcProblemFormulation:generalMpc}
    \begin{align}
        \label{eq:sec:mpcProblemFormulation:generalMpc:cost}
        \min_{u_{\text{c}},\,\vec{\epsilon}}&\ J(u_{\text{c}},\vec{\epsilon})= \int_{t_i}^{t_{i}+t_{\text{hor}}} \underbrace{\norm{\vec{r}-\vec{p}(\theta)}_2 + P(\vec{\epsilon})}_{=l(\vx,\theta,u_{\text{c}},\vec{\epsilon})}\dt\\
        \label{eq:sec:mpcProblemFormulation:generalMpc:rhsState}
        \text{s.t.}\quad 
        &\vxd=\vf(\vx,u_{\text{c}}), \quad \vx(t_i)=\vxhat_{i}\\
        \label{eq:sec:mpcProblemFormulation:generalMpc:rhsPathParameter}
        &\dot{\theta}=q(\theta,\vx),\quad \ \ \, \theta(t_i)=\hat{\theta}_i\\
        \label{eq:sec:mpcProblemFormulation:generalMpc:collisionAvoidanceConstraint}
        &h_{m}(\vx)\leq \epsilon_m, \quad  \ \ \ \: \, m=1,\hdots,M\\
        \label{eq:sec:mpcProblemFormulation:generalMpc:boxConstraints}
        &u_{\text{c}}^-\leq u_{\text{c}}\leq u_{\text{c}}^+\\
        \label{eq:sec:mpcProblemFormulation:generalMpc:rateConstraints}
        &\abs{\dot{u}_{\text{c}}}\leq \dot{u}_{\text{c}}^{\text{max}}\\
        \label{eq:sec:mpcProblemFormulation:generalMpc:epsilonPositive}
        &\vec{\epsilon}\geq\vec{0}
    \end{align}
\end{subequations}
where $t_i$ is the current iteration time, $t_{\text{hor}}$ is the prediction horizon. The control input $u_c$ and the slack variables $\vec{\epsilon}\inRn{M}$ constitute the decision variables. In this problem, \eqref{eq:sec:mpcProblemFormulation:generalMpc:cost} is the cost to be minimized and $l(\cdot)$ denotes the running costs, where $\norm{\vec{r}-\vec{p}(\theta)}_2$ gives the Euclidean distance or cross track error between the vessel position $\vec{r}$ and the reference path $\vec{p}(\theta)$ as a function of the path parameter $\theta$ which is typically parameterized as the arc length of the path. Each successive minimization is subject to the dynamic equality constraints \eqref{eq:sec:mpcProblemFormulation:generalMpc:rhsState}, given by \eqref{eq:sec:surfaceVessel:generalVesselOde}, and the path following timing law \eqref{eq:sec:mpcProblemFormulation:generalMpc:rhsPathParameter}, which dictates the propagation in time of the path parameter $\theta$ and constitutes a degree of freedom in the path-following MPC. The initial path parameter $\hat{\theta}_i$ depends on the vessel state $\hat{\vec{x}}_i$ and is calculated before each iteration. The inequality constraints  \eqref{eq:sec:mpcProblemFormulation:generalMpc:collisionAvoidanceConstraint} represent state constraints and, in this particular context, obstacle constraints for which three different implementations are considered. Using the soft-constrained approach, any violation of \eqref{eq:sec:mpcProblemFormulation:generalMpc:collisionAvoidanceConstraint} is penalized in the cost function by means of the term $P(\vec{\epsilon})=\transpose{\vec{\epsilon}}S\vec{\epsilon}+\transpose{\vec{s}}\vec{\epsilon}$, which can be used to avoid feasibility issues see, e.g., \cite{Scokaert1999a}. The inequality constraints \eqref{eq:sec:mpcProblemFormulation:generalMpc:boxConstraints} represent box constraints on the control input, and \eqref{eq:sec:mpcProblemFormulation:generalMpc:rateConstraints} enforces input rate constraints. The inequality \eqref{eq:sec:mpcProblemFormulation:generalMpc:epsilonPositive} ensures that the slack variables are non-negative.
See also \cite{Borelli2020}, \cite{Scokaert1999a} for further information on the soft-constrained approach and \cite{Faulwasser2016} for theoretical results on nonlinear model predictive path following control. In the following, a detailed view of the constraints in \eqref{eq:sec:mpcProblemFormulation:generalMpc} with regard to a dual collision-avoidance path-following MPC for an underactuated surface vessel is given.

\subsection{Timing Law}
The timing law \eqref{eq:sec:mpcProblemFormulation:generalMpc:rhsPathParameter} for the path-following MPC is chosen  based on \cite{Paliotta2019} to be
\begin{align}  
  \label{eq:sec:mpcProblemFormulation:timingLaw}
    \dot{\theta} &=\underbrace{u_0\left[1-\sigma\tanh\left(\frac{\transpose{\vec{e}}(\theta,\vx)\vec{e}(\theta,\vx)}{L_{\text{pp}}^2}\right)\right]}_{=q(\theta,\vx)},
\end{align}
where $u_0$ is the service speed of the vessel, $\sigma\in(0,1)$ is a tuning parameter, $L_{\text{pp}}$ is the length of the vessel, and 
\begin{align*}
    \vec{e}(\theta,\vec{x}) = \vec{r}- \vec{p}(\theta)
\end{align*}
is the vector defining the cross track error $\norm{\vec{e}}_2$ between the vessel position $\vec{r}$ and the current reference point on the line segment defined by two consecutive waypoints $\vec{r}^{\text{wp}}_i$ and $\vec{r}^{\text{wp}}_{i+1}$, i.e.
\begin{align*}
    \vec{p}(\theta)&=\vec{r}^{\text{wp}}_i + \alpha\left(\theta\right)(\vec{r}^{\text{wp}}_{i+1}-\vec{r}^{\text{wp}}_i),
\end{align*}
where $i$ denotes the current waypoint and
\begin{align*}
    \alpha\left(\theta\right)=\frac{\theta-\sum_{j=1}^ {i-1}\norm{\vec{r}^{\text{wp}}_{j+1}-\vec{r}^{\text{wp}}_j}_2}{\norm{\vec{r}^{\text{wp}}_{i+1}-\vec{r}^{\text{wp}}_i}_2}.
\end{align*}
As can be seen, $\theta\in\left[0,\,\sum_i \norm{\vec{r}^{\text{wp}}_i-\vec{r}^{\text{wp}}_{i+1}}_2\right]$ is parameterized as the arc length of the reference path. The timing law \eqref{eq:sec:mpcProblemFormulation:timingLaw} ensures that the propagation satisfies $q(\theta,\vx)>0, \  \forall t$ and that it propagates slowly if the distance between the vessel position and the reference point is large and vice versa. See \cite{Faulwasser2016} for further information on path-following MPC.

\subsection{Obstacle constraints}\label{subsec:obstacleConstraints}
In the following, the obstacle constraints \eqref{eq:sec:mpcProblemFormulation:generalMpc:collisionAvoidanceConstraint} are considered. First, the often-used ellipsoidal representation of obstacles is presented, where the controlled vessel's geometry can not be taken into account explicitly. Subsequently, the proposition \eqref{eq:sec:dualSignedDistance:resultsBorelli} based on \cite{Borelli2020} is considered along with the proposed representation given by \eqref{eq:sec:dualSignedDistance:signedDistanceResult}, 
which both assume the controlled vessel and all obstacles to be of polyhedral shape.
\subsubsection{Ellipsoidal formulation}
An ellipsoidal obstacle can be represented with $\mathcal{O}_m^\circ=\{ \vx\inRn{2}\vert (\vx-\vr_m)^\top M_m(\vx-\vr_m)\leq 1\}$, where $\vr_m$ is the position of the $m$-th obstacle's CO, $M_m=\text{diag}\{ \nicefrac{1}{(a_m+d_{\text{safe}})^2},\, \nicefrac{1}{(b_m+d_{\text{safe}})^2} \}\inRnm{2}{2}$, and $a_m,b_m$ define the ellipse's major and minor axis length, respectively\footnote{Note that for $a_m=b_m$, $\mathcal{O}_m^\circ$ represents a circle with radius $R_m=a_m=b_m$.}. This formulation can be translated directly into the obstacle constraint defined by \eqref{eq:sec:mpcProblemFormulation:generalMpc:collisionAvoidanceConstraint}, i.e.,
\begin{align}
    \label{eq:sec:mpcProblemFormulation:ellipsoidalConstraints}
    1-(\vr-\vr_m)^\top M_m(\vr-\vr_m) \leq  \epsilon_m,\ \ m=1,\hdots,M.
\end{align}
With this, \eqref{eq:sec:introduction:collisionAvoidanceCondition} can be ensured, but the explicit geometry of the controlled vessel can not be taken into account. Note that a logarithmic relaxation of \eqref{eq:sec:mpcProblemFormulation:ellipsoidalConstraints} often yields a better scaled problem, see, e.g., \cite{Bitar2019b}.
\subsubsection{Dual polyhedral formulation according to \cite{Borelli2020}}
The conditions \eqref{eq:sec:dualSignedDistance:resultsBorelli} described in \cite{Borelli2020} can be translated into a set of obstacle constraints
\begin{subequations}
    \label{eq:sec:mpcProblemFormulation:borelliConstraints}
    \begin{align}
    \label{eq:sec:mpcProblemFormulation:borelliCollisionAvoidance}
    d_{\text{safe}}+\transpose{\vec{\lambda}_m}\vec{d}+\transpose{\vmum}\vbm &\leq \epsilon_m\\
    \label{eq:sec:mpcProblemFormulation:borelliNorm}
    \norm{\transpose{\Am}\vmum}&=1\\
    \label{eq:sec:mpcProblemFormulation:borelliConsistency}
    \transpose{C}\vec{\lambda}_m+\transpose{A_m}\vmum&=\vec{0}\\
    \label{eq:sec:mpcProblemFormulation:borelliNonnegativity}	
    \vmum,\vlambdam&\geq\vec{0}
    \end{align}
\end{subequations}
for all $m=1,\hdots,M$. Therefore, additional dual obstacle decision variables \mbox{$ [\transpose{\vmu_{1}} \ \hdots\ \transpose{\vmu_{M}}]$} and dual controlled vessel decision variables \mbox{$[ \transpose{\vlam_{1}} \ \hdots\ \transpose{\vlam_{M}}]$} need to be taken into account as part of \eqref{eq:sec:mpcProblemFormulation:generalMpc}.
In comparison to \eqref{eq:sec:mpcProblemFormulation:generalMpc}, the constraint \eqref{eq:sec:mpcProblemFormulation:borelliCollisionAvoidance} corresponds to the collision-avoidance inequality constraints \eqref{eq:sec:mpcProblemFormulation:generalMpc:collisionAvoidanceConstraint} in the general setup. The constraints \eqref{eq:sec:mpcProblemFormulation:borelliNonnegativity} constitute additional box constraints. Additionally, the dual formulation \eqref{eq:sec:mpcProblemFormulation:borelliConstraints} relies on norm equality constraints \eqref{eq:sec:mpcProblemFormulation:borelliNorm}. Moreover, consistency equality constraints \eqref{eq:sec:mpcProblemFormulation:borelliConsistency} need to be imposed.
\subsubsection{Proposed dual polyhedral formulation}\label{sec:alternative_formulation}
Taking into account Proposition \ref{prop:sd_Vk_Okm}, the constraints for a dual collision-avoidance formulation can be written as
\begin{subequations}
    \label{eq:sec:mpcProblemFormulation:proposedConstraints}
    \begin{align}
        d_{\text{safe}}+\frac{\transpose{\vec{\lambda}_m}\vec{d}+\transpose{\vmum}\vbm}{\norm{\transpose{\Am}\vmum}}&\leq \epsilon_m,\\
        \transpose{C}\vec{\lambda}_m+\transpose{A_m}\vmum&=\vec{0},\\
        \label{eq:mpc_Vk_Okm_box_constraints_compact}	
        \vmum,\vlambdam&\geq\vec{0}
    \end{align}
\end{subequations}
for all $m=1,\hdots,M$. In this way, the number of necessary equality constraints can be reduced compared to \eqref{eq:sec:mpcProblemFormulation:borelliConstraints} which leads to a less complex problem. The decision variables, inequality and box constraints remain unaltered compared to \eqref{eq:sec:mpcProblemFormulation:borelliConstraints}.
\section{Simulation results}\label{sec:results}
\begin{figure}[t!]
    \centering
    \vspace*{.2cm}%
    \subcaptionbox{Simulation snapshot for iteration $i=79$ using polygons for obstacles and the controlled vessel showing the vessel's initial pose (magenta), current and predicted poses (green points with blue polygons), predicted reference points (violet), reference path (black lines), waypoints (black points), and the past vessel trajectory (dash-dotted orange). One of the two obstacles (black polygon with red safety margin) is currently blocking the path. \label{fig:overview_a}}[\columnwidth]
    {
        \includegraphics{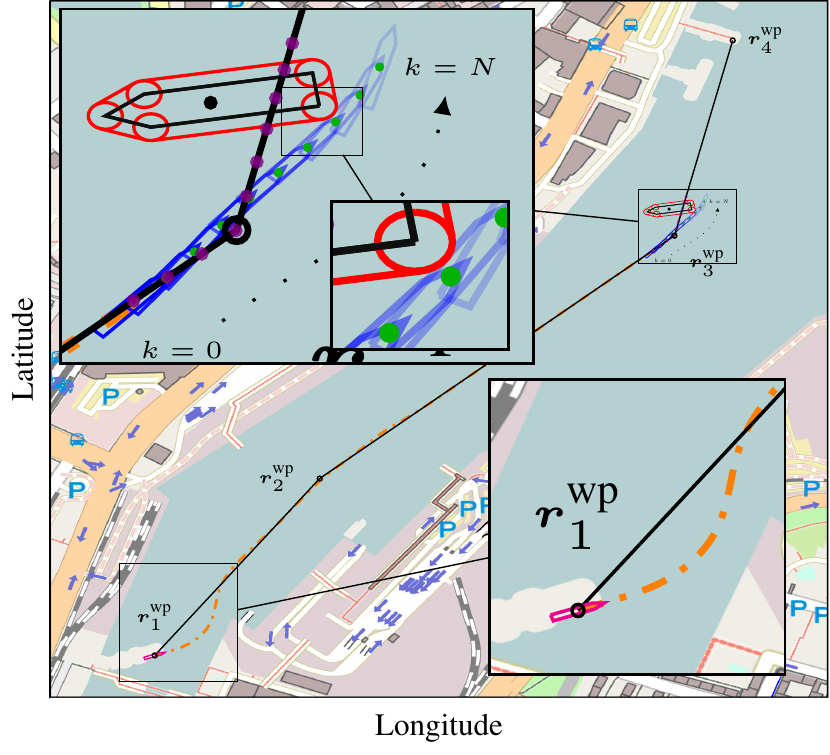}
    }\\
    \subcaptionbox{Input $\delta$ (blue) with constraints (dashed red) for entire simulation.\label{fig:overview_b}}[\columnwidth]
    {
        \includegraphics{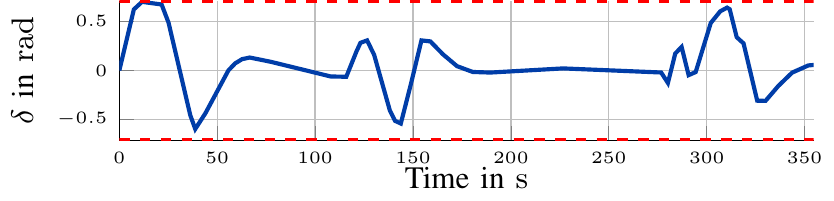}
    }%
    \captionsetup{belowskip=-.7cm}
    \caption{Overview of the simulation scenario and path-following performance.}\label{fig:overview}
\end{figure}
 A mariner class vessel is  used for the simulation and is scaled to a length of $L_{\text{pp}}=\SI{25}{\meter}$ in order to represent a typical ferry size from the Kiel bay area. The rudder angle is constrained to $\delta\in[ -\SI{40}{\degree}, \SI{40}{\degree} ]$ and permits rudder-rate changes up to $\dot{\delta}_{\text{max}}=\SI[per-mode=symbol]{5}{^\circ\per\second}$. The constant service speed of the scaled model reads $u_0\approx \SI[per-mode=symbol]{3}{\meter\per\second}$. Four reference waypoints are chosen such that a transition from the pier at Kiel central station to Kiel Reventlou Bridge is achieved. Two obstacles with $L_1=L_2=5$ are chosen from AIS data, where the first is a static obstacle and the second passes the reference path during the maneuver. The safety margin is set to $d_{\text{safe}}=\SI{4}{\meter}$. The tuning parameter of the timing law \eqref{eq:sec:mpcProblemFormulation:timingLaw} is set to $\sigma=0.2$. The initial state reads $\hat{\vx}_0=\transpose{[\SI{54.3154}{\degree N}, \SI{10.135592}{\degree E}, \SI{1.2}{\radian},  \SI[per-mode=symbol]{0}{\meter\per\second}, \SI[per-mode=symbol]{0}{\meter\per\second}, \SI[per-mode=symbol]{0}{\radian\per\second}]}$ and each problem \eqref{eq:sec:mpcProblemFormulation:generalMpc} is solved using a direct multiple shooting method with $N=10$ discretization steps and sample time $\Delta t=\SI{3.5}{\second}$. All dynamic constraints and costs are integrated numerically using a trapezoidal scheme. See, e.g., \cite{Betts1998} for an overview of different numerical methods for optimal control problems. The implementation utilizes \matlab with \snopt as NLP solver, see \cite{Gill2005}, on an Intel$^{\text{\textregistered}}$ Core\textsuperscript{TM} i5-6200U CPU with 2.30\si{\giga\Hz} clock speed. A snapshot of the path-following simulation is depicted in  Fig. \ref{fig:overview_a}, where a good path-following performance can be observed. Note that only one of the obstacles is seen in this snapshot. Further, note that a comparison between the original and proposed implementation shows no significant difference in the solution trajectory and is therefore omitted. The realized input for the entire simulation is depicted in Fig. \ref{fig:overview_b}. The vessel is subject to disturbances induced by wind similar to \cite{Helling2020} and an Extended Kalman Filter is used to estimate the system states. The different evading behaviors depending on the obstacle formulation are shown in Fig. \ref{fig:evading} for the obstacle that has already been passed in the snapshot in Fig. \ref{fig:overview_a}. As can be seen in Fig. \ref{fig:evadingPolyhedron}, the polyhedral formulation allows for a smaller cross track error when evading the obstacle while in Fig. \ref{fig:overview_a}, the ellipsoidal obstacle requires the controlled vessel to move further away from the path. While this issue could be overcome using a true ellipse instead of a circle, the ellipsoidal formulation can not take into account the polyhedral geometry of the controlled vessel (although it is shown in the figure). This results in parts of the vessel being inside the safety margin of the ellipsoidal obstacle while the dual polyhedral formulation circumvents this issue, see Fig. \ref{fig:evading}. The latter is therefore more suited for applications in confined environments.
\begin{figure}[t!]
    \vspace*{.2cm}%
    \subcaptionbox{Ellipsoidal formulation.\label{fig:evadingCircle}}[.5\columnwidth]
    {
        \includegraphics{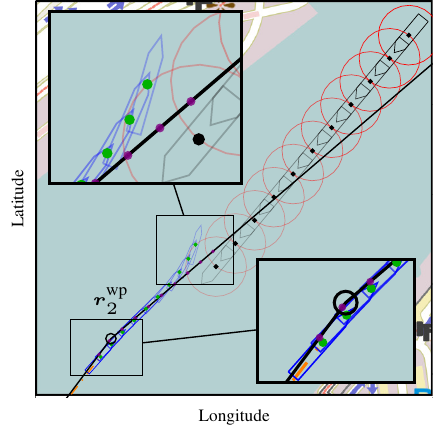}
    }%
    \subcaptionbox{Dual polyhedral formulation.\label{fig:evadingPolyhedron}}[.5\columnwidth]
    {
        \includegraphics{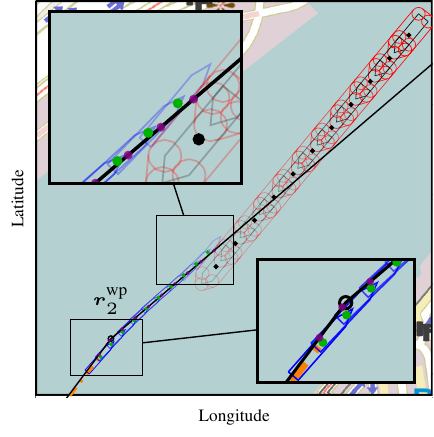}
    }
    \captionsetup{belowskip=-.7cm}
    \caption{Simulation snapshot at iteration $i=33$ showing the evading behavior of the ellipsoidal obstacle formulation (left) and the dual polyhedral formulation (right) with predicted vessel poses (green points and blue polygons), safety margins (red) and predicted reference points (violet).}\label{fig:evading}
\end{figure}

In Tab. \ref{tab:computation_times}, the average computation time per iteration can be compared along with the resulting nonlinear static program (NLP) dimensions concerning the different obstacle formulations discussed in Sec. \ref{subsec:obstacleConstraints}. Compared to the ellipsoidal case, the computation time increases for the dual polyhedral case due to the increased number of decision variables and constraints. Especially incorporating the explicit vessel geometry increases the problem dimension. The proposed dual formulation outperforms the original implementation without a significant change in the resulting trajectories rendering the proposed method a promising candidate for optimization-based collision-avoidance strategies in confined environments.
\begin{table}[htb]
    \renewcommand{\arraystretch}{1.2}
    \caption{Comparison of general problem dimensions regarding obstacle constraints and decision variables, as well as average computation times (per MPC iteration) for different obstacle implementations, where P stands for the proposed polyhedral implementation.}
    \label{tab:computation_times}
    \centering
    \noindent
    \begin{tabulary}{\columnwidth}{m{1cm} m{1.6cm} ||C |C C}
        \toprule
        \multicolumn{2}{l||}{\multirow{2}{*}{Obstacle formulation}} & {Ellipsoidal} & \multicolumn{2}{c}{Polyhedral} \\
        & & \eqref{eq:sec:mpcProblemFormulation:ellipsoidalConstraints} & \eqref{eq:sec:mpcProblemFormulation:borelliConstraints} & \multicolumn{1}{l}{\eqref{eq:sec:mpcProblemFormulation:proposedConstraints}}\\
        \midrule
        \midrule
        \multirow{2}{=}{decision variables} &{dual obstacle} & $0$ & \multicolumn{2}{c}{ $(N+1)\sum L_m$}\\
        &{dual vessel}   & $0$ & \multicolumn{2}{c}{ $(N+1)ML$}\\
        \midrule
        \multirow{3}{=}{constraints}& norm                        & $0$ & $(N+1)M$ &0\\
        & con\-sis\-ten\-cy           & $0$ & \multicolumn{2}{c}{ $2(N+1)M$} \\
        & co\-lli\-sion-avoi\-dan\-ce & \multicolumn{3}{c}{-----------\ $(N+1)M$\ -----------}\\
        \midrule
        \multicolumn{2}{c||}{$\overline{t}_{\text{cpu}}$ in \si{\second}} & 0.22 &{0.51}&{0.42}\\
        \bottomrule
    \end{tabulary}
\end{table}
\vspace*{-.3cm}
\section{Conclusion}\label{sec:conclusion}
A path-following MPC with a dual implementation based on \cite{Borelli2020} of the signed distance function is presented. The signed distance function is used to formulate a soft-constrained collision-avoidance optimal control problem for convex polyhedral obstacles, where the controlled vessel's geometry is taken into account. In this contribution, an alternative formulation is introduced resulting in a less complex NLP. The theoretical findings are validated using a simulative study with AIS data from the Kiel bay area. The proposed formulation yields a more efficient implementation compared to the problem in \cite{Borelli2020}. The simulation compares the dual formulations with an ellipsoidal obstacle formulation. Future studies will focus on reducing problem complexity of the dual approach, especially since the number of dual decision variables increases disproportionately with an increasing number of obstacles, as well as experimental verification of the proposed concepts. 
\vspace*{-.25cm}

\bibliographystyle{IEEEtran}
\bibliography{root}
\end{document}